\documentclass[journal]{IEEEtran}

\usepackage{cite}
\usepackage{balance}
\ifCLASSINFOpdf
\usepackage[pdftex]{graphicx}
\usepackage{epstopdf}
\usepackage{subcaption}
\usepackage{pgfplots}
\usepackage{amsmath,stmaryrd,pict2e,picture}
\usepackage{float}
\else
\fi
\usepackage{amsmath,amssymb,amsthm}
\usepackage{mathrsfs}
\usepackage{algorithm, algorithmic}

\newtheorem{prop}{Proposition}

\begin{document}
\title{Hybrid Sparse Array  Beamforming Design for General Rank Signal Models}
\author{Syed~A.~Hamza,~\IEEEmembership{Student Member,~IEEE}
        and~Moeness~G.~Amin,~\IEEEmembership{Fellow,~IEEE}
\thanks{This work is supported by National Science Foundation (NSF) award no. AST-1547420. This paper was presented in part at the 2019 IEEE International Conference on Acoustics, Speech and Signal Processing, Brighton, U.K., May 2019 \cite{8682266} and 2018 IEEE  Radar Conference, Oklahoma City, USA, April 2018 \cite{8378759}. (\textit{Corresponding author: Syed A. Hamza}.)}
\thanks{The authors are with the Center for Advanced Communications (CAC),  College of
	Engineering, Villanova University, PA 19085-1681 USA (e-mails: shamza@villanova.edu, moeness.amin@villanova.edu).}
}
\maketitle
\begin{abstract}
The paper considers sparse array design for receive beamforming achieving maximum signal-to-interference plus noise ratio (MaxSINR) for both single point source and multiple point sources, operating in an interference active environment. Unlike existing sparse design methods which either deal with structured environment-independent or non-structured environment-dependent arrays, our method is a hybrid approach and seeks a full augumentable array that optimizes beamformer performance. This approach proves important for limited aperture that constrains the number of possible uniform grid points for sensor placements. The problem is formulated as quadratically constraint quadratic program (QCQP), with the cost function penalized with weighted $l_1$-norm squared of the beamformer weight vector. Simulation results are presented to show the effectiveness of the proposed algorithms for array configurability in the case of both single and general rank signal correlation matrices. Performance comparisons among the proposed sparse array, the commonly used  uniform arrays, arrays obtained by other design methods, and arrays designed without the augmentability constraint are provided.
\end{abstract}
\IEEEpeerreviewmaketitle
\section{Introduction}
Sparse array design through sensor selection reduces system receiver overhead by lowering the hardware costs and processing complexity. It finds applications in sensor signal processing for communications, radar, sonar, satellite navigation,  radio telescopes, speech enhancement  and ultrasonic imaging \cite{710573, 6477161, 1428743, 4058261,  4663892,   6031934}. One primary goal in these applications is to determine sensor locations to achieve optimality for some pre-determined performance criteria. This optimality includes minimizing the mean radius of the confidence ellipsoid associated with the estimation error covariance matrix  \cite{4663892}, and lowering the  Cramer Rao bound (CRB)   for angle estimation in direction finding problem \cite{6981988}. The receiver performance then depends largely on the operating environment, which may change according to the source and interference signals and locations.  This is in contrast to sparse arrays whose configurations follow certain formulas and seek to attain high extended aperture co-arrays. The driving objective, in this case,  is to enable direction of arrival (DOA) estimation of more sources than physical sensors. Common examples are structured arrays such as nested and coprime arrays \cite{1139138, 5456168, 7012090}.
 
Sparse array design typically involves the selection of a subset of uniform grid points for sensor placements. For a given number of sensors, it is often assumed that the number of grid points, spaced by half wavelength,  is unlimited. However, in many applications, there is a constraint on  the spatial extent of the system aperture. In this case,  a structured array, in seeking to maximize the number of spatial autocorrelation lags,  may find itself placing  sensors beyond the available physical aperture. The problem then becomes that of dual constraints, one relates to the number of sensors, and the other to the number of grid-points.

 With a limited aperture constraint invoked, few sensors may in fact be sufficient to produce a desirable filled structured co-array, even with narrowband assumption and without needing wideband or multiple frequencies \cite{7106511}. In this case, any additional sensors, constitute a surplus that  can be utilized to meet an environment-dependent performance criterion, such as maximum signal-to-interference and noise ratio (SINR).  Thereby, one   can in essence reap the benefits of structured and non-structured arrays. This paradigm calls for a new aperture design approach that strives to provide filled co-arrays and, at the same time, be environment-sensitive.  This hybrid design approach is the core contribution of this paper.

Sparse sensor  design has  thoroughly been studied  to economize the  receive beamformer \cite{1144829, 1144961, 1137831, Lo1966ASO, 1534, 97353, 5680595, doi:10.1029/2010RS004522,  299602, 775291, 6663667, article1, 6145602, 6822510}.  However, in contrast to MaxSINR design,  the main focus of  the efforts, therein, was in achieving  desirable beampattern characteristics with nominal sidelobe levels, since the sparse beamformer  is susceptible to high  sidelobe levels. For example, an array thinning design was proposed for sidelobe minimization in \cite{1534}  by starting from a fully populated array and sequentially  removing sensors in a systematic manner. Instead, the sparse array design presented in \cite{97353} to  optimize the peak sidelobe level   involves a  joint design of sensor locations and their corresponding beamforming weights. A beampattern matching design explained in \cite{5680595} can  effectively recover  sparse topologies  through an iterative cyclic approach.  Additionally, global optimization tools such as Genetic Algorithms/Simulated Annealing and  convex relaxation schemes based on re-weighted $l_1$-norm minimization have  been  rigorously exploited in sensor selection problem  for synthesizing a user-specified receive beampattern response \cite{299602, 775291, 6663667, article1, 6145602, 6822510}.

In environment-dependent array design, signal power estimation and enhancement in an interference active environment  has a direct bearing on improving   target detection and localization for radar   signal processing,   increasing   throughput or  channel capacity for MIMO wireless communication systems, and enhancing   resolution capability in  medical imaging \cite{Goldsmith:2005:WC:993515, Trees:1992:DEM:530789, 1206680}. It is noted that with sparse array, the commonly used Capon beamforming must not only find the optimum weights but also the optimum array configuration. This is clearly an entwined optimization problem, and requires finding maximum SINR over all possible sparse array configurations. Maximum signal to noise ratio (MaxSNR) and MaxSINR have been shown to yield significantly efficient beamforming with  performance depending mainly on the positions of the sensors as well as the locations of sources in the field of view (FOV) \cite{6774934,  1634819, 6714077}.      \par

In this paper, we  consider a bi-objective optimization problem, namely achieving the filled co-array  and maximizing the SINR. The proposed technique enjoys  key advantages as compared to state-of-the-art sparse aperture design, namely, (a) It does not require any \textit{a priori} knowledge of the  jammers directions of arrival and their respective power which is implicitly assumed in  previous contributions \cite{8061016, 8313065, 8036237}. As such, it is possible to directly work  on the received data   correlation matrix   (b) It extends to spatial spread sources in a straightforward way.  \par
\begin{figure}[!t]
	\centering
	\includegraphics[width=3.35 in, height=2.3 in]{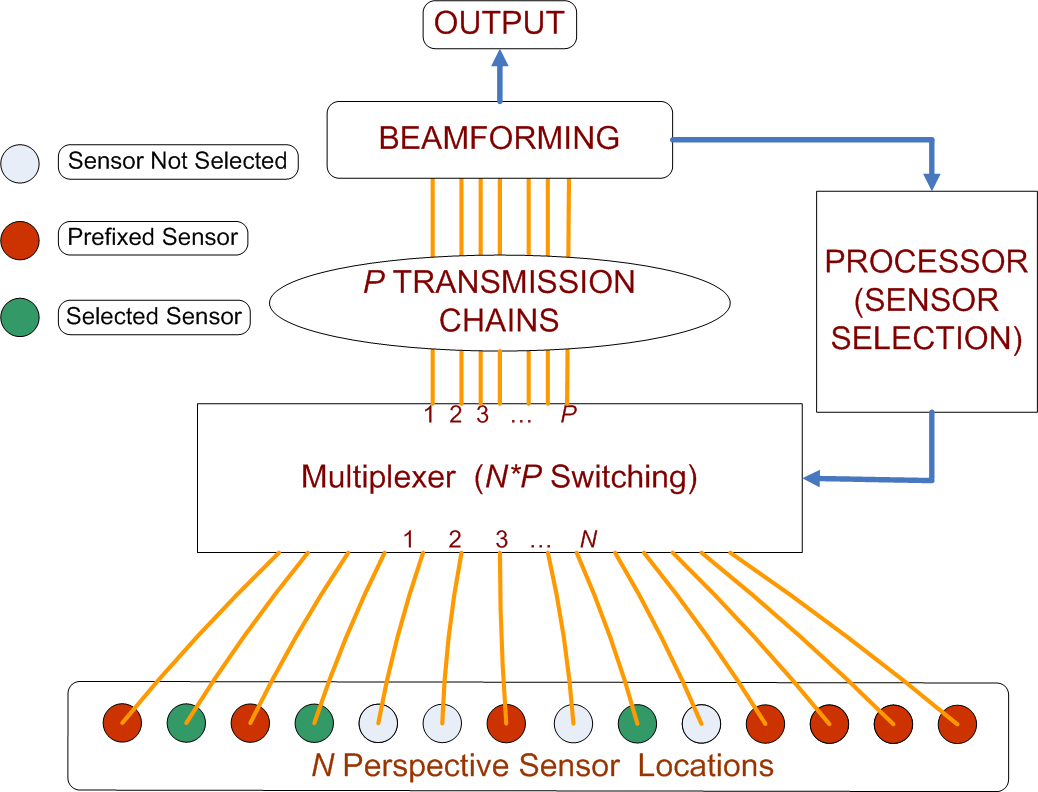}
	\caption{Block diagram of adaptive switched sensor beamformer}
	\label{Block_diagram}
\end{figure}
 The proposed hybrid approach first determines a prefixed sparse array that results in a filled   co-array with minimum number of sensors. This prefixed configuration could be a minimum redundancy array (MRA) \cite{1139138}, nested or coprime array configuration that fills the aperture under consideration with minimal sensors, allowing maximum degrees of freedom for SINR maximization. This prefixed sensor configuration can  be  achieved by an optimization problem involving the minimum number of sensors spanning a pre-determined aperture. However, for the scope of this paper, the prefixed configuration is  set by MRA or other structured arrays.   The remaining sensors after forming the prefixed array are  utilized to maximize the  SINR.  The cascade nature of the proposed hybrid approach is relatively simpler than the  ultimate design approach that produces the optimum filled sparse array that maximizes SINR. Environment-dependent array design lowers the hardware complexity by reducing the expensive transmission chains through sensor switching as shown in the  block diagram   in  Fig. \ref{Block_diagram}. The proposed hybrid approach, however,  has an added advantage of offering a simplified sensor switching in time-varying environment.  This is attributed to   large number of fixed location sensors which would always remain  non-switched, irrespective of the sources and interferences in the FOV. 
 
  The proposed hybrid approach is  particularly permissive as the number $N$ of possible sensor locations increases.  To further clarify, it is noted that sparse arrays having $N$ available sensors can typically span a filled array aperture of the order of $\mathcal{O}(N(N-1)/2)$ \cite{5456168}; conversely, given an aperture spanning  $N$ possible sensor locations, only  $\mathcal{O}(N^{1/2})$ sensors are sufficient to synthesize a fully augmentable array design. This emphasizes the fact that as the possible aperture size increases, then relatively  few sensors are required to meet the full augmentability condition,  leaving more degrees of freedom to optimize  for SINR enhancement. The hybrid  approach also lends itself to  more desirable  beampattern characteristics  by  maintaining  minimum spacing between sensor elements. It is important to note  that having fully augmentable arrays not only provide    the benefits of simplified sensor switching and improved identifiability of large number of sources,  but also they ensure  the availability of full array data covariance matrix essential to carry  optimized SINR configuration \cite{709534}, \cite{Abramovich:1999:PTC:2197953.2201686}.  Therefore, the proposed  simplified hybrid  sensor switching architecture ensures the knowledge of global data statistics at all times,   in contrast to  previous efforts in \cite{7096707, 6333987, 8234698} that sort to optimize data dependent microphone placement viz a viz transmission power. The  proposed methodology therein   targets a different objective function and  primarily relies  on  local heuristics. In this  case,  sensor switching comes with an  additional  implementation overhead, in an  attempt to recursively match the  performance offered by the knowledge of global statistics.     \par

We consider the problem of MaxSINR sparse arrays with limited aperture for both single and higher rank signal correlation matrices. The case of single rank correlation matrix arises when there is one desired source signal in the FOV, whereas the case of higher rank signal model occurs for  spatially spread source.   The   problem is posed as optimally selecting $P$ sensors out of $N$ possible equally spaced grid points. Maximizing SINR amounts to maximizing the  principal eigenvalue of the product of  the inverse of data correlation matrix and the desired source correlation matrix \cite{1223538}. Since it   is an NP hard optimization problem, we pose this problem as QCQP with weighted $l_1$-norm squared to promote sparsity. The re-weighted $l_{1}$-norm squared relaxation is effective for reducing the required sensors and   minimizing the transmit power for multicast beamforming  \cite{ 6477161}. We propose a modified re-weighting matrix  based iterative approach to control the sparsity of the optimum weight vector so that  $P$ sensor fully augmentable hybrid array is finally selected. This modified regularization re-weighting matrix  based approach  incorporates the prefixed structured array   assumption in our design and works by minimizing the objective function   around the presumed prefixed array.   \par  
The rest of the paper is organized as follows:  In the next section, we  state the problem formulation for  maximizing the  output SINR under general rank signal correlation matrix. Section \ref{Optimum sparse array design} deals with the optimum sparse array design by  semidefinite relaxation (SDR) and proposed  modified re-weighting  based iterative algorithm of finding $P$ sensor fully augmentable  hybrid  sparse array design. In section \ref{Simulations}, with the aid of number of  design examples,  we demonstrate the usefulness of fully augmentable arrays achieving MaxSINR and highlight the effectiveness of the proposed methodology  for sparse array design. Concluding remarks follow at the end.
\section{Problem Formulation} \label{Problem Formulation}
Consider  $K$  desired sources and $L$ independent interfering source signals impinging on a  linear array with $N$ uniformly placed sensors. The baseband signal received at the array  at  time instant $t$ is then given by; 
\begin {equation} \label{a}
\mathbf{x}(t)=   \sum_{k=1}^{K} (\alpha _k(t)) \mathbf{s}( \theta_k)  + \sum_{l=1}^{L} (\beta _l(t)) \mathbf{v}( \theta_l)  + \mathbf{n}(t),
\end {equation}
where, $\mathbf{s} ({\theta_k})$  and $\mathbf{v} ({\theta_l})$ $\in \mathbb{C}^N$ are  the corresponding steering vectors    respective to directions  of arrival, $\theta_k$ or $\theta_l$, and are defined  as follows;
\vspace{+2mm}
\begin {equation}  \label{b}
\mathbf{s} ({\theta_k})=[1 \,  \,  \, e^{j (2 \pi / \lambda) d cos(\theta_k)  } \,  . \,   . \,  . \, e^{j (2 \pi / \lambda) d (N-1) cos(\theta_k)  }]^T.
\end {equation}
The inter-element spacing  is denoted  by $d$, ($\alpha _k(t)$, $\beta _l(t))$ $\in \mathbb{C}$  denote the complex amplitudes of the incoming baseband signals \cite{trove.nla.gov.au/work/15617720}. The additive Gaussian noise $\mathbf{n}(t)$ $\in \mathbb{C}^N$ has a  variance of  $\sigma_n^2$ at the receiver output.
The received signal vector $\mathbf{x}(t)$   is   combined linearly by the $N$-sensor  beamformer that strives to maximize the output SINR. The output signal $y(t)$ of the optimum beamformer for maximum SINR is given by \cite{1223538}, 
\begin {equation}  \label{c}
y(t) = \mathbf{w}_o^H \mathbf{x}(t),
\end {equation}
where $\mathbf{w}_o$ is the solution of the  optimization problem given below;
\begin{equation} \label{d}
\begin{aligned}
\underset{\mathbf{w} \in \mathbb{C}^N}{\text{minimize}} & \quad   \mathbf{w}^H\mathbf{R}_{s^{'}}\mathbf{w},\\
\text{s.t.} & \quad     \mathbf{ w}^H\mathbf{R}_{s}\mathbf{ w}=1.
\end{aligned}
\end{equation}
For statistically independent signals, the desired source correlation matrix is given by, $\mathbf{R}_s= \sum_{k=1}^{K} \sigma^2_k \mathbf{s}( \theta_k)\mathbf{s}^H( \theta_k)$, where, $ \sigma^2_k =E\{\alpha _k(t)\alpha _k^H(t)\}$. Likewise,  we have the  interference and noise correlation matrix $\mathbf{R}_{s^{'}}= \sum_{l=1}^{L} (\sigma^2_l \mathbf{v}( \theta_l)\mathbf{v}^H( \theta_l)$) + $\sigma_n^2\mathbf{I}_{N\times N}$,    with $ \sigma^2_l =E\{\beta _l(t)\beta_l^H(t)\}$ being the power of the $l$th interfering source. The problem in (\ref{d}) can  be written equivalently by replacing $\mathbf{R}_{s^{'}}$ with the received data covariance matrix, $\mathbf{R_{xx}}=\mathbf{R}_s+ \mathbf{R}_{s^{'}}$  as follows \cite{1223538},
\begin{equation} \label{e}
\begin{aligned}
\underset{\mathbf{w} \in \mathbb{C}^N }{\text{minimize}} & \quad   \mathbf{ w}^H\mathbf{R_{xx}}\mathbf{ w},\\
\text{s.t.} & \quad     \mathbf{ w}^H\mathbf{R}_{s}\mathbf{ w} \geq 1.
\end{aligned}
\end{equation}
It is noted that the  equality constraint in (4) is  relaxed in (5) due to  the inclusion of  the constraint as part of the objective function, and as such, (5)  converges to  the equality constraint. Additionally, the optimal solution in (5) is invariant up to  uncertainty of the absolute powers of the sources of interest.  Accordingly, the relative power profile of the sources of interest would suffice.  For a single desired point source, this implies that only the knowledge of the DOA of the desired source is sufficient  rather than the exact knowledge of the desired source correlation matrix. Similarly, neither the source power nor the average power of the scatterers is required in (5) for spatially spread sources when the spatial channel model, such as the Gaussian or circular, is assumed \cite{656151}.  However, in practice, these assumptions can deviate from the actual received data statistics and hence the discrepancy is typically mitigated, to an extent,  by preprocessing  the received data correlation matrix through diagonal loading or tapering the  correlation matrix \cite{1206680}.

     There exists a  closed form solution of the above  optimization problem  and  is given by $\mathbf{w}_o = \mathscr{P} \{  \mathbf{R}_{s^{'}}^{-1} \mathbf{R}_s  \}=\mathscr{P} \{  \mathbf{R_{xx}}^{-1} \mathbf{R}_s  \}$. The operator $\mathscr{P} \{. \}$  computes the principal eigenvector of the input matrix. Substituting $\mathbf{w}_o$ into  (\ref{c})  yields the corresponding optimum output SINR$_o$;
\begin{equation}  \label{f}
\text{SINR}_o=\frac {\mathbf{w}_o^H \mathbf{R}_s \mathbf{w}_o} { \mathbf{w}_o^H \mathbf{R}_{s^{'}} \mathbf{w}_o} = \Lambda_{max}\{\mathbf{R}^{-1}_{s^{'}} \mathbf{R}_s\}.
\end{equation}
 This shows that the optimum output SINR$_o$ is given by the maximum eigenvalue ($\Lambda_{max}$) associated with   the product of  the inverse of interference plus noise correlation matrix  and the desired source correlation matrix. Therefore, the performance of the optimum beamformer for maximizing the output SINR is directly related to the desired and interference plus noise correlation matrix. It is to be noted that the  rank of the desired source signal correlation matrix equals $K$, i.e. the cardinality of the desired sources.

\section{Optimum sparse array design}  \label{Optimum sparse array design}
The problem of locating the maximum  principal eigenvalue among all the   correlation matrices associated with $P$ sensor selection is a combinatorial optimization problem. The  constraint optimization (\ref{e}) can be re-formulated for optimum sparse array design by incorporating  an additional constraint on the cardinality of  the weight vector;   
\begin{equation} \label{a2}
\begin{aligned}
\underset{\mathbf{w \in \mathbb{C}}^{N}}{\text{minimize}} & \quad  \mathbf{ w}^H\mathbf{R_{xx}}\mathbf{w},\\
\text{s.t.} & \quad   \mathbf{ w}^H\mathbf{R}_{s}\mathbf{ w}\geq1,  \\
& \quad   ||\mathbf{w}||_0=P.
\end{aligned}
\end{equation}
Here,  $||.||_0$ determines the cardinality of the weight   vector $\mathbf{w}$.  We  assume that we have an estimate of all the filled co-array  correlation lags corresponding to the correlation matrix of the  full  aperture array. The problem expressed in  (\ref{a2})  can  be relaxed  to induce the sparsity in the beamforming weight vector $\mathbf{w}$ without placing a hard constraint on the specific cardinality of $\mathbf{w}$, as follows \cite{doi:10.1002/cpa.20132};
\begin{equation} \label{b2}
\begin{aligned}
\underset{\mathbf{w \in \mathbb{C}}^{N}}{\text{minimize}} & \quad   \mathbf{ w}^H\mathbf{R_{xx}}\mathbf{w} + \mu(||\mathbf{w}||_1),\\
\text{s.t.} & \quad    \mathbf{ w}^H\mathbf{R}_{s}\mathbf{ w} \geq1.  \\
\end{aligned}
\end{equation}
Here, $||.||_1$  is the  sparsity inducing $l_1$-norm and $\mu$ is a parameter to control the desired sparsity in the solution. Even though the relaxed problem expressed in  (\ref{b2}) is not exactly similar to that of  (\ref{a2}), yet it is well known that  $l_1$-norm  regularization has been an effective tool for recovering sparse solutions in many diverse formulations \cite{Bruckstein:2009:SSS:1654028.1654037,5651522, doi:10.1190/1.1440921}. The  problem in $(\ref{b2})$  can be  penalized instead  by the weighted $l_1$-norm function  which is a well known sparsity promoting formulation \cite{Candès2008},
\begin{equation} \label{c2}
\begin{aligned}
\underset{\mathbf{w \in \mathbb{C}}^{N}}{\text{minimize}} & \quad   \mathbf{ w}^H\mathbf{R_{xx}}\mathbf{w} + \mu(||(\mathbf{b}^i\circ|\mathbf{w}|)||_1),\\
\text{s.t.} & \quad    \mathbf{ w}^H\mathbf{R_{s}}\mathbf{ w} \geq1.  \\
\end{aligned}
\end{equation}
where, \textquotedblleft$\circ$\textquotedblright \, denotes the element wise product, \textquotedblleft$|.|$\textquotedblright \, is the modulus operator and $\mathbf{b}^i \in \mathbb{R}^N$  is the regularization re-weighting vector  at  the $i$th iteration. Therefore, $(\ref{c2})$ is the sequential optimization methodology, where the regularization re-weighting vector $\mathbf{b}^i$ is typically chosen as an inverse function of the beamforming weight vector obtained at the previous iteration. This, in turn,  suppresses the sensors corresponding to  smaller beamforming weights, thereby encouraging sparsity in an iterative fashion.    The weighted $l_1$-norm function   in $(\ref{c2})$  is replaced by the  $l_1$-norm  squared  function which does not  alter the regularization property of  the weighted  $l_1$-norm function \cite{6477161},
\begin{equation} \label{d2}
\begin{aligned}
\underset{\mathbf{w \in \mathbb{C}}^{N}}{\text{minimize}} & \quad   \mathbf{ w}^H\mathbf{R_{xx}}\mathbf{w} + \mu(||(\mathbf{b}^i\circ|\mathbf{w}|)||^2_1),\\
\text{s.t.} & \quad    \mathbf{ w}^H\mathbf{R}_{s}\mathbf{ w} \geq1.  \\
\end{aligned}
\end{equation}
The semidefinite  formulation (SDP) of the above problem can then be realized  by re-expressing the quadratic form, $ \mathbf{ w}^H\mathbf{R_{xx}}\mathbf{w}= $Tr$(\mathbf{ w}^H\mathbf{R_{xx}}\mathbf{w})= $Tr$(\mathbf{R_{xx}}\mathbf{w}\mathbf{ w}^H)=$ Tr$(\mathbf{R_{xx}}\mathbf{W})$, where Tr(.) is the trace of the matrix. Similarly, the regularization term $||(\mathbf{b}^i\circ|\mathbf{w}|)||^2_1=(|\mathbf{w}|^T\mathbf{b}^i)( (\mathbf{b}^i)^T|\mathbf{w}|)=\mathbf{|w|}^T\mathbf{B}^i\mathbf{|w|}=$Tr$(\mathbf{B}^i\mathbf{|W|})$. Here,  $\mathbf{W}=\mathbf{w}\mathbf{w}^H$ and $\mathbf{B}^i=\mathbf{b}^i(\mathbf{b}^i)^T$ is the regularization re-weighting matrix at the $i$th iteration. Utilizing these quadratic expressions in (\ref{d2})  yields the following problem \cite{Bengtsson99optimaldownlink, 8378759, 6477161},
\begin{equation} \label{e2}
\begin{aligned}
\underset{\mathbf{W \in \mathbb{C}}^{N\times N}, \mathbf{\tilde{W} \in \mathbb{R}}^{N\times N}}{\text{minimize}} & \quad   \text{Tr}(\mathbf{R_{xx}}\mathbf{W}) + \mu \text{Tr}(\mathbf{B}^i\mathbf{\tilde{W}}),\\
\text{s.t.} & \quad    \text{Tr}(\mathbf{R}_s\mathbf{W}) \geq 1,  \\
& \quad \mathbf{\tilde{W}} \geq |\mathbf{W}|,\\
& \quad   \mathbf{W} \succeq 0, \, \text{Rank}(\mathbf{W})=1.
\end{aligned}
\end{equation}

The function \textquotedblleft$|.|$\textquotedblright \, returns the  absolute values of the entries of the matrix, \textquotedblleft$\geq$\textquotedblright \, is the element wise comparison and \textquotedblleft$\succeq$\textquotedblright \, denotes the  generalized  matrix inequality. The auxiliary matrix $\mathbf{\tilde{W}} \in \mathbb{R}^{N\times N}$ implements the weighted  $l_1$-norm  squared regularization along with the re-weighting matrix $\mathbf{B}^i$.   The  rank constraint in  (\ref{e2}) is non convex and therefore need to be removed. The rank relaxed approximation works well for the underlying problem. In case, the solution matrix is not rank $1$, we can resort to randomization to harness  rank $1$ approximate solutions \cite{5447068}. Alternatively, one could  minimize the nuclear norm of $\mathbf{W}$,  as a  surrogate for $l_1 $-norm in the case of   matrices, to induce sparsity in the eigenvalues of $\mathbf{W} $ and promote rank one solutions  \cite{Recht:2010:GMS:1958515.1958520, Mohan:2012:IRA:2503308.2503351}. The  resulting rank relaxed semidefinite program (SDR) is  given  by;
\begin{equation} \label{f2}
\begin{aligned}
\underset{\mathbf{W \in \mathbb{C}}^{N\times N}, \mathbf{\tilde{W} \in \mathbb{R}}^{N\times N}}{\text{minimize}} & \quad   \text{Tr}(\mathbf{R_{xx}}\mathbf{W}) + \mu \text{Tr}(\mathbf{B}^i\mathbf{\tilde{W}}),\\
\text{s.t.} & \quad    \text{Tr}(\mathbf{R}_s\mathbf{W}) \geq 1,  \\
& \quad \mathbf{\tilde{W}} \geq |\mathbf{W}|,\\
& \quad   \mathbf{W} \succeq 0.
\end{aligned}
\end{equation}
 In general, QCQP  is NP hard and cannot  be solved in polynomial time.  The formulation in  (\ref{f2})  is clearly convex, in terms of unknown matrices, as all the  other correlation matrices involved  are guaranteed to be positive semidefinite. The sparsity parameter $\mu$ largely determines the cardinality of the solution beamforming weight vector.   To ensure $P$ sensor selection, appropriate value of $\mu$ is typically found by carrying a binary search over the probable range of $\mu$. After achieving the desired cardinality, the reduced size thinned correlation matrix $\mathbf{R_{xx}}$ is formed corresponding to the non-zero values of $\mathbf{\tilde{W}}$. The reduced dimension SDR is now solved with setting $\mu=0$, yielding optimum beamformer $\mathbf{w}_o=\mathscr{P} \{  \mathbf{W} \}$. 
\subsection{Fair gain beamforming}
The optimization in   (\ref{f2}) strives to incorporate the signal from all the directions of interest while optimally removing the interfering signals. To achieve this objective, the optimum sparse array may  show leaning towards a certain source of interest, consequently, not offering fair gain towards all sources.  In an effort to promote equal gain towards all sources, we put a separate constraint on the power towards all desired sources as follows;
\begin{equation} \label{g2}
\begin{aligned}
\underset{\mathbf{W \in \mathbb{C}}^{N\times N}, \mathbf{\tilde{W} \in \mathbb{R}}^{N\times N}}{\text{minimize}} & \quad   \text{Tr}(\mathbf{R_{xx}}\mathbf{W}) + \mu \text{Tr}(\mathbf{B}^i\mathbf{\tilde{W}}),\\
\text{s.t.} & \quad    \text{Tr}(\mathbf{R}_k\mathbf{W}) \geq 1, \, \, \, \, \forall k\in (1,2,3 ... K)   \\
& \quad \mathbf{\tilde{W}} \geq |\mathbf{W}|,\\
& \quad   \mathbf{W} \succeq 0.
\end{aligned}
\end{equation}
 Here,  $\mathbf{R}_k= \mathbf{s}( \theta_k)\mathbf{s}^H( \theta_k)$ is the rank $1$ covariance matrix associated with the source at DOA $( \theta_k)$.   However, the above SDR can   be solved to an arbitrary small accuracy  $\zeta$, by employing  interior point methods  involving   the   worst case complexity of  $\mathcal{O}\{$max$(K,N)^4N^{(1/2)}\log(1/\zeta)\}$   \cite{5447068}.
\begin{algorithm}[t!] \label{algorithm}
	
	\caption{Proposed algorithm to achieve desired cardinality of optimal weight vector $\mathbf{w}_o$.}
	
	\begin{algorithmic}[]
		
		\renewcommand{\algorithmicrequire}{\textbf{Input:}}
		
		\renewcommand{\algorithmicensure}{\textbf{Output:}}
		
		\REQUIRE Data correlation matrix $\mathbf{R_{xx}}$, $N$, $P$, look direction DOA's $\theta_k$,  hybrid selection vector $\mathbf{z}$. \\
		
		\ENSURE  $P$ sensor beamforming weight vector $\mathbf{w}_o$,   \\
				    Initialize  $\epsilon$.\\
				 	Initialize $\mu_{lower}$, $\mu_{upper}$ (Initializing lower and upper limits of sparsity parameter range for	binary search for desired cardinality $P$)\\

		   \textbf{FSDR:} Initialize  $\mathbf{B=zz}^T$. \\
		   \textbf{NFSDR:} For optimum array design without the augmentability constraint, initialize  $\mathbf{z}$  to be all ones vector, $\mathbf{B=zz}^T$ (all ones matrix). \\
		   \textbf{Perturbed-NFSDR:} Locate the sensor $i$ such that, if not selected, results in the minimum compromise of the objective function. Perturb $\mathbf{z}$ at position $i$, $\mathbf{z}(i)=\mathbf{z}(i)+\gamma$, afterwards calculating $\mathbf{B=zz}^T$. \\
		 \WHILE {(Cardinality of $\mathbf{w}_o$ $\neq$ $P$)}
		\STATE  Update $\mu$ through binary search.
		\FOR {(Typically requires five to six iterations)}
			\STATE   Run the  SDR of  (\ref{f2}) or (\ref{g2}) (Fair gain case).
		   \STATE   Update the regularization weighting matrix $\mathbf{B}$ according to (\ref{i2}). \\ 
		\ENDFOR
         	\ENDWHILE
		
		\STATE After achieving the desired cardinality, run SDR for reduced size correlation matrix corresponding to nonzero values of $\mathbf{\tilde{W}}$  and $\mu=0$,  yielding, $\mathbf{w}_o=\mathscr{P} \{  \mathbf{W} \} $. 
		
		\RETURN $\mathbf{w}_o$
		
	\end{algorithmic}
	\label{algorithm}
\end{algorithm}
\subsection{Modified re-weighting for fully augmentable hybrid array}\label{Modified Re-Weighting}
 For the case without the full augmentability constraint the   regularization re-weighting matrix  $\mathbf{B}$ is initialized unweighted i.e. by all ones matrix and the $m,n$th element of \textbf{B} is iteratively updated as follows \cite{Candès2008},
\begin{equation} \label{h2}
\mathbf{B}_{m,n}^{i+1}=\frac{1}{|\mathbf{W}_{m,n}^i|+\epsilon}.
\end{equation}
The parameter $\epsilon$ avoids the unwanted case of division by zero, though its choice is fairly independent to the performance of the iterative algorithm but at times very small values of $\epsilon$ can result in the algorithm getting trapped in the local minima.  For the hybrid array design, we initialize the re-weighting matrix instead as an outer product of hybrid selection vector $\mathbf{z}$.  The hybrid selection vector $\mathbf{z}$ is an $N$ dimensional vector containing binary entries of zero and one,  where, zeros correspond to the pre-selected sensors and ones correspond to the remaining sensors to be selected. Hence, the  cardinality  of $\mathbf{z}$ is equal to the  difference of the total number of available sensors and the number of pre-selected sensors. This modified re-weighting approach ensures  that the sensors corresponding to the pre-selected configuration is  not penalized as part of the regularization, hence,  $\mathbf{B=zz}^T$, thrives solutions that incorporate the pre-selected array topology. The modified penalizing weight update for the hybrid array design can be expressed as;
\begin{equation} \label{i2}
\mathbf{B}^{i+1}=(\mathbf{zz}^T) \varoslash ({|\mathbf{W}^i|+\epsilon}).
\end{equation}
The symbol \textquotedblleft$\varoslash$\textquotedblright \, denotes element wise division. For the hybrid design, (15) is proposed with appropriate selection of  $\mathbf{z}$, as explained above, and hereafter referred to as the Fixed SDR (FSDR).  The array designed without the augmentability consideration is the special case of (15) with $\mathbf{z}$ being an all ones vector and  the algorithm is subsequently regarded as the Non-Fixed SDR (NFSDR).   The pseudo-code  for controlling the sparsity of the optimal weight vector $\mathbf{w}_o$ is summarized in   $\bf{Algorithm \, \, \ref{algorithm}}$.
\subsection{Symmetric arrays}\label{Symmetric arrays}
 The solution of the  NFSDR formulation is penchant for symmetric arrays in the case of  symmetric initialization vector   $\mathbf{z}$.  The plausible explanation is  as follows.  We first show that the  beamforming weights which maximizes the output SINR for symmetric sparse array topologies are conjugate symmetric w.r.t. the array center.
\begin{prop}
The conjugate symmetry of the optimal weight vector   holds for centro-symmetric sparse array configurations in case of the general rank desired source  model.
\end{prop}
\begin{proof}
(Refer to the \ref{Appendix1}ppendix for the proof.) 
\end{proof}
      We observe that the regularized cost function does not  invoke sparsity until after the first  few initial iterations. Consequently,   the initial solutions of the semidefinite program has symmetric coefficients as the NFSDR seeks near optimal solutions which  are  analytically shown to be conjugate symmetric.  Moreover, the iterative sparsity enhancing formulation introduces sparsity by penalizing the beamforming weight vector according to  (\ref{i2}), where, it only accounts the magnitude of the beamforming weights. Therefore, at each iteration the regularization re-weighting matrix $\mathbf{B}$ happens to penalize the solution weight vector in a symmetric fashion around the array center.  Thus, the   iterative NFSDR sparse solution favors   symmetric configurations by discarding corresponding symmetric  sensors simultaneously. Though, the symmetric configuration can  be suitable for certain applications \cite{8448767}, and can have desirable performance, yet,  it reduces the available degrees of freedom. Therefore, to avoid curtailing  the available degrees of freedom,  we   perturb the re-weighting regularization matrix   $\mathbf{B}$ at the initial iteration, as follows.   From $N$ prospective locations,  find the   sensor position,  which if not selected, results in the least compromise of  the objective function performance. Corresponding to the aforementioned  position,  set the  regularization weight to be relatively high through perturbation by parameter $\gamma$. By so doing, we resolve the issues arising from the  symmetric regularization re-weighting matrix. This modified algorithm is henceforth referred to as the perturbed-NFSDR and is detailed in $\bf{Algorithm \, \, \ref{algorithm}}$. 
\begin{figure}[t]
	\centering
	\includegraphics[width=3.38 in, height=2.5 in]{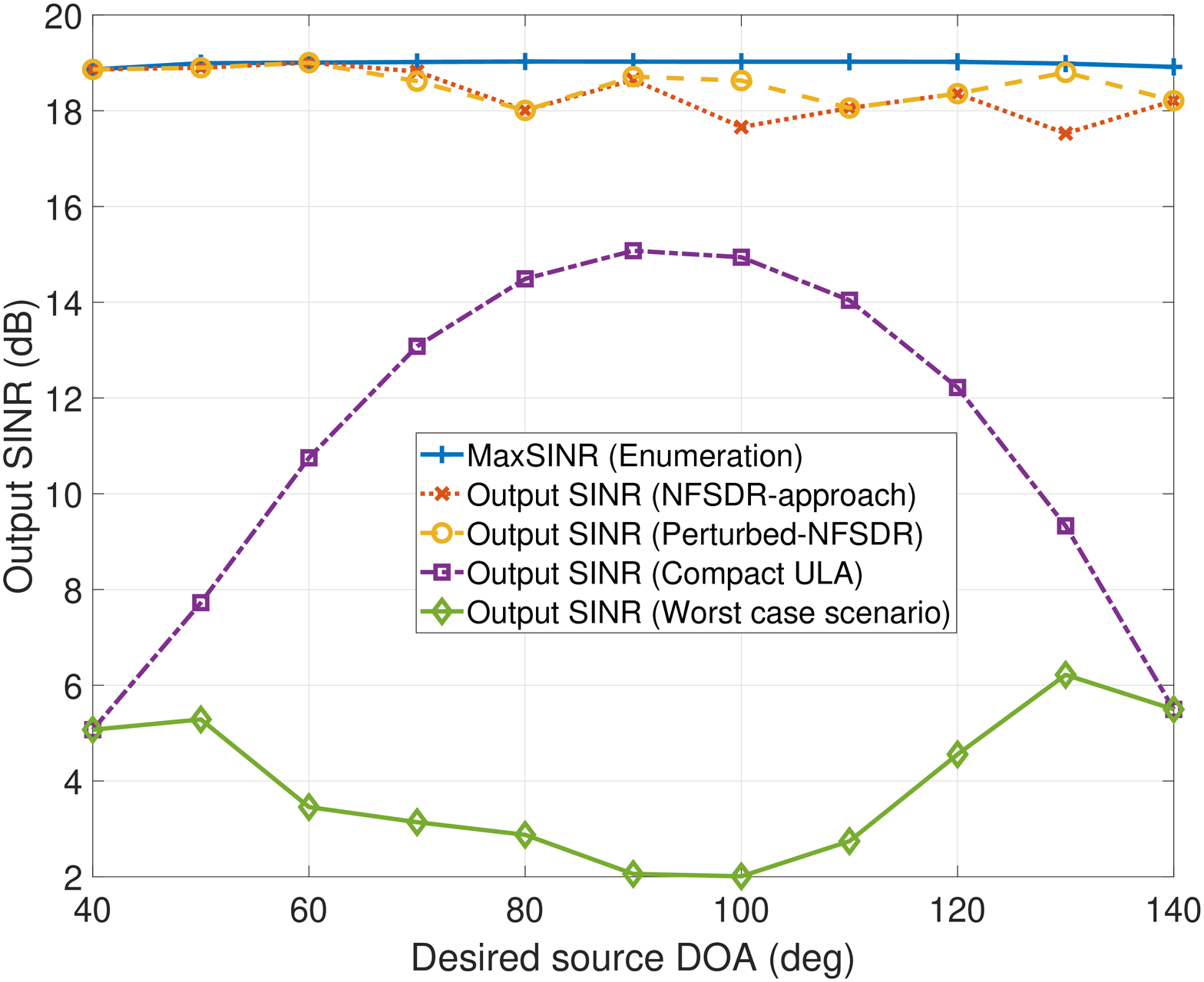}
	\caption{Output SINR for different array topologies}
	\label{scan}
\end{figure}
\section{Simulations} \label{Simulations}
In this section, we show the effectiveness of the proposed techniques for the sparse array design for  MaxSINR.  We initially examine the  proposed approach for  array configurability   by considering  arbitrary arrays  without  the augmentability constraint. In the later examples,  we  demonstrate the effectiveness  of fully augmentable  hybrid sparse array design through linear and 2D arrays. We focus on the EM modality, and as such we use antennas for sensors.

\subsection{Single point source} \label{Single point source}
We select $P=8$ antennas from $N=16$  possible equally spaced locations with  inter-element spacing of $\lambda/2$. Figure \ref{scan} shows the output SINR for different  array configurations for  the case  of single desired point source with its DOA varying from   $40^0$ to $140^0$. The interfering signals are located at $20^0$ and $\pm10^0$ degree apart  from the desired source angle. To explain this scenario, suppose that the desired source is at $60^0$, we consider the respective directions of arrival of the three interfering signals   at $40^0$, $50^0$ and $70^0$. The SNR of the desired signal is $10$ dB, and the  interference to noise ratio (INR)  is set to $10$ dB for each scenario. The input SINR is $-4.9$ dB. The upper and lower limit of the sparsity parameter $\mu$ is set to $1.5$ and $0.01$ respectively, $\gamma=0.05$  and $\epsilon=0.1$. From  the  Fig. \ref{scan}, it is evident that the NFSDR-approach performs  close to the performance of the optimum array found by exhaustive search ($12870$ possible configurations), which has very high computational cost attributed to expensive singular value decomposition (SVD) for each enumeration. Moreover, the perturbed-NFSDR algorithm results in comparable or better performance. Except for the slightly lower performance at the desired source of DOA of $70^0$, we observe that for the desired source of DOA at $90^0$, $100^0$ and $130^0$, the perturbed-NFSDR recovers a sparse array with better performance than the NFSDR-approach. For the other DOAs, the perturbed-NFSDR recovers the same symmetric configuration as that recovered by the NFSDR-approach.  This emphasizes that the perturbed-NFSDR does not  eliminate the possibility of  symmetric solutions and optimizes over  both the symmetrical and unsymmetrical array configurations.    On average, the proposed algorithms  takes six to seven  iterations to converge to the optimum antenna locations; hence, offering  considerable savings in the computational cost. It is of interest to compare  the optimum sparse array performance with that of compact uniform linear array (ULA). It can be seen from    Fig. \ref{scan}, that the optimum sparse array offers considerable SINR advantage over the compact  ULA for all   source  angles of arrival. The  ULA performance degrades  severely when the source of interest is more towards the array end-fire location.  In this case, the   ULA fails to resolve and cancel the strong interferers as they are located close to the desired source. 
\begin{figure}[t]
	\centering
	\includegraphics[width=3.38 in, height=2.50 in]{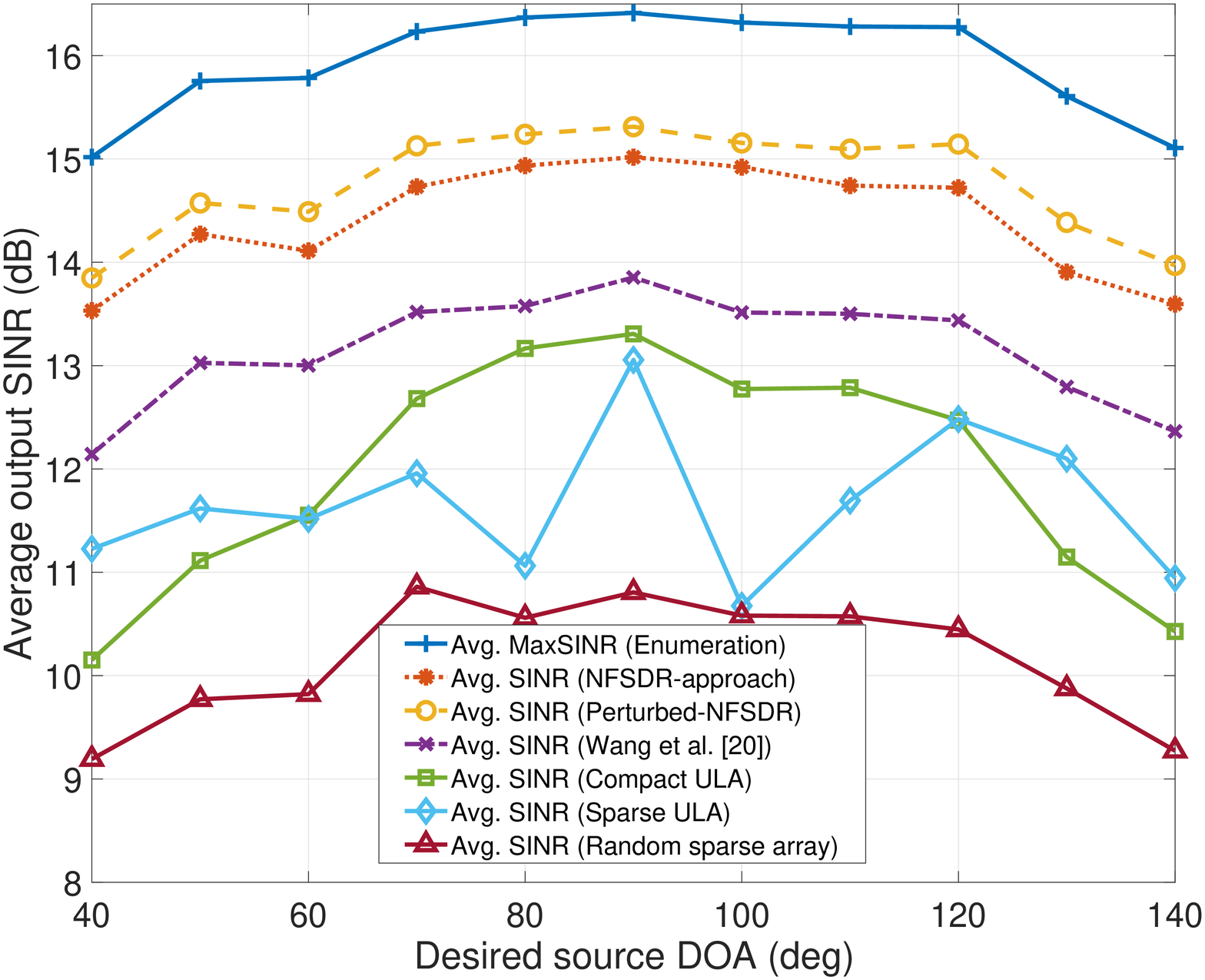}
	\caption{Average Output SINR for different array topologies over $6000$ Monte Carlo trials}
	\label{scan_monte}
\end{figure}

\begin{figure}[!t]
	\centering
	\begin{subfigure}[b]{0.5\textwidth}
		\includegraphics[width=8.9cm, height=1.2cm]{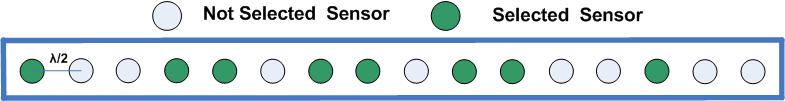}
		\caption{}
		\label{best_90_0}
	\end{subfigure}
	\begin{subfigure}[b]{0.5\textwidth}
		\includegraphics[width=8.9cm, height=0.75cm]{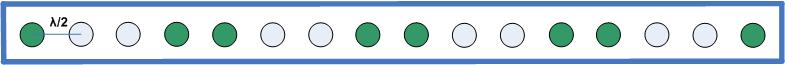}
		\caption{}
		\label{op_algo_90_0}
	\end{subfigure}
	\begin{subfigure}[b]{0.5\textwidth}
		\includegraphics[width=8.9cm, height=0.75cm]{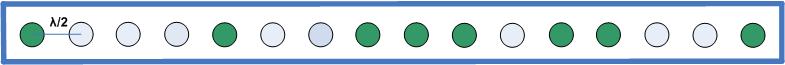}
		\caption{}
		\label{op_symalgo_90_0}
	\end{subfigure}
	\begin{subfigure}[b]{0.5\textwidth}
		\includegraphics[width=8.9cm, height=0.75cm]{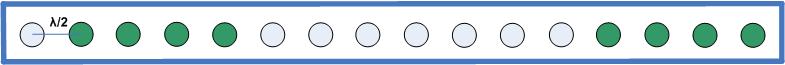}
		\caption{}
		\label{worst_90_0}
	\end{subfigure}
	\caption{Array configurations obtained  for the  point source at the array broadside (a) Optimum  (Enumeration)  (b) NFSDR-approach  (c) Perturbed-NFSDR  (d) Worst performing array configuration  }
	\label{point}
\end{figure}

For the case of the desired source  at the array broadside,  the maximum output SINR  of the optimum array found through enumeration (Fig. \ref{best_90_0})  is $19$ dB.  The optimum array  design obtained through the  NFSDR-approach yields an output SINR of $18.6$ dB, which is  $0.4$ dB less than the corresponding SINR of the optimum array found through exhaustive search.  The broadside source arrays are shown in the  Fig. \ref{point} (where green-filled circle indicates antenna present whereas gray-filled circle indicates antenna absent).  The sparse array recovered through NFSDR-approach is clearly a symmetric configuration (Fig. \ref{op_algo_90_0}). Figure  \ref{op_symalgo_90_0} shows the sparse array found after addressing the symmetry bias by the  approach explained in Section \ref{Symmetric arrays}. The SINR for this non-symmetric configuration is $18.7$ dB and is suboptimal merely by  $0.3$ dB.   It is worth noticing  that  the worst performing sparse array configuration (Fig. \ref{worst_90_0}) comparatively engages larger  array aperture than the optimum array found through enumeration (Fig. \ref{best_90_0}), yet it  has an output SINR as low as  $2.06$ dB. This  emphasizes the fact that if an arbitrary  sparse array structure is employed, it could degrade the performance catastrophically irrespective of the occupied aperture  and could perform far worst than the compact ULA, which offers modest output SINR of  $15.07$ dB  for the  scenario under consideration. 
\subsubsection{Monte Carlo Simulation} \label{Monte1}
To thoroughly examine the performance of the proposed algorithms under  random  interfering environments, we perform 6000 Monte Carlo simulations.  For this purpose, the desired source DOA is fixed with SNR of $10$ dB, and eight  interferences are generated  which are uniformly distributed anywhere  from $20^0$ to $160^0$.  The INRs of these sources are  uniformly drawn from $10$ dB to $15$ dB. We choose 8 antennas out of 16 possible locations. The upper and lower limit of the sparsity parameter $\mu$ is set to 3 and 0.01 respectively, $\gamma=0.1$  and $\epsilon=0.05$. The performance curves are shown in  Fig. \ref{scan_monte} for the desired source fixed at $11$ different DOAs varying from $40^0$ to $140^0$. On  average, the proposed perturbed-NFSDR  algorithm consistently provided superior SINR performance. However,  this performance is around $1.2$ dB suboptimal than the average SINR computed through enumeration. The average SINR performance of the perturbed-NFSDR algorithm is around $0.35$ dB better than the  proposed NFSDR-approach. This is because the degrees of freedom are limited by the inherent array symmetry enforced by the re-weighted optimization scheme. The performances of the proposed algorithms are compared with the  design methodology proposed in \cite{8313065}, which relies on the \textit{a priori} knowledge of the interference steering vectors and respective powers. It is noted  that in the underlying scenario the design in \cite{8313065} is more than $1$ dB suboptimal than the proposed algorithms and around $2$ dB suboptimal as compared to the performance upper bound. The algorithm in \cite{8313065} relies on successive linear approximation of the objective function as opposed to the  quadratic implementation of the SDR, thereby suffering  in performance. The SINR performances  for the compact ULA, sparse ULA and randomly employed sparse topology are also shown in the Fig. \ref{scan_monte}, further highlighting the utility of sparse array design.
\begin{figure}[t]
	\centering
	\begin{subfigure}[b]{0.5\textwidth}
		\includegraphics[width=8.9cm, height=1.2cm]{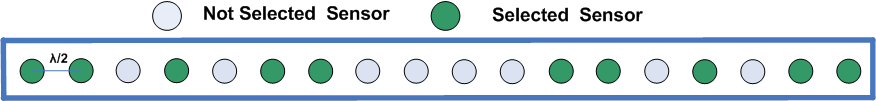}
		\caption{}
		\label{op_alg_45_5}
	\end{subfigure}
	\begin{subfigure}[b]{0.5\textwidth}
		\includegraphics[width=8.9cm, height=0.8cm]{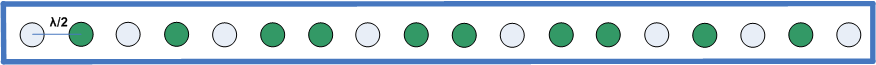}
		\caption{}
		\label{op_45_5}
	\end{subfigure}
	\begin{subfigure}[b]{0.5\textwidth}
		\includegraphics[width=8.9cm, height=1.2cm]{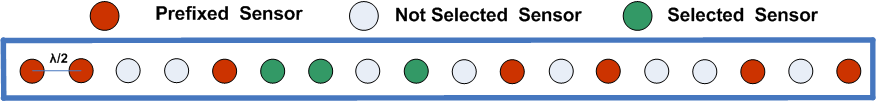}
		\caption{}
		\label{op_alg_multi}
	\end{subfigure}
	\caption{(a) Antenna array multiple sources (NFSDR-approach)   (b)   Fair gain $10$ element antenna array (NFSDR-approach)    (c)  Hybrid  $10$ antenna array  for multiple desired sources (FSDR)    }
	\label{op_c_co_c}
\end{figure}
\begin{figure}[t]
	\centering
	\includegraphics[width=3.38 in, height=2.50 in]{beampatterneune.eps}
	\caption{Beampattern  for  multiple point sources}
	\label{multi_bp}
\end{figure}
\begin{figure}[t]
	\centering
	\begin{subfigure}[b]{0.5\textwidth}
		\includegraphics[width=8.9cm, height=1.2cm]{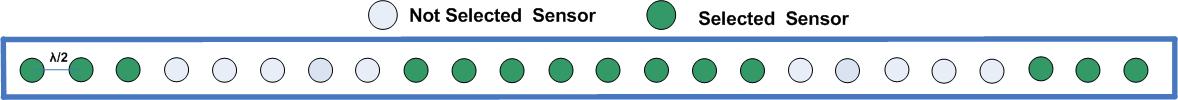}
		\caption{}
		\label{op_1D}
	\end{subfigure}
	\begin{subfigure}[b]{0.5\textwidth}
		\includegraphics[width=8.9cm, height=1.2cm]{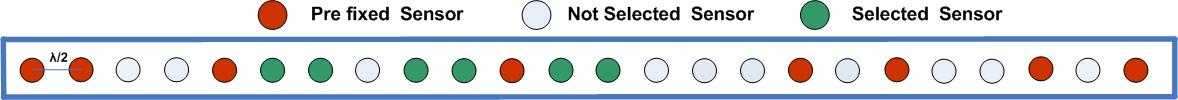}
		\caption{}
		\label{mm1_1D-hybrid}
	\end{subfigure}
	\begin{subfigure}[b]{0.5\textwidth}
		\includegraphics[width=8.9cm, height=0.79cm]{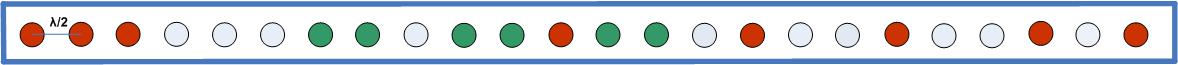}
		\caption{}
		\label{mm2_1D-hybrid}
	\end{subfigure}
	\caption{(a) $14$ element antenna  array (NFSDR-approach)  (b)    Hybrid  $14$ antenna sparse array ($8$ prefixed, $6$ selected through FSDR)   (c)  Hybrid $14$ antenna sparse array ($8$ prefixed, $6$ selected through FSDR)}
	\label{1D-hybrid}
\end{figure}

\subsection{ Multiple point sources} \label{Multiple point sources}
For the multiple point sources scenario,  consider three desired signals  impinging from DOAs $40^0$,  $65^0$ and $90^0$ with SNR of $0$ dB each. Unlike the example in \ref{Single point source}, we set four strong interferers with INR of $30$ dB are operational at DOAs $50^0$,  $60^0$,  $120^0$  and $150^0$. In so doing, we  analyze the robustness of the proposed scheme under very low input SINR of $-36.02$ dB.  We select $10$ antennas out of $18$ available slots. The optimum array recovered through convex relaxation is shown in Fig. \ref{op_alg_45_5}. This configuration results with an  output SINR of $11.85$ dB  against SINR of $12.1$ dB for the optimum configuration found through enumeration. For the fair gain beamforming, we apply the optimization of   (13) and the  array configuration for MaxSINR for  the fair gain beamforming is shown in Fig \ref{op_45_5}. The output SINR for the fair beamforming case is $11.6$ dB which is slightly less than the optimum array without the fair gain consideration ($11.85$ dB). However, the advantage of  fair beamforming is well apparent from the beampatterns in both cases as shown in Fig \ref{multi_bp}, where the  gain towards the source at $65^0$ is around $4.24$ dB higher than the case of optimum array without the fair gain consideration. The maximum gain deviation for the fair gain case is $3.5$ dB vs. $8$ dB variation without the fair gain consideration.   The SINR  of compact ULA is  compromised more than $3$ dB   as compared to the optimum sparse array (Fig. \ref{op_alg_45_5}) obtained through the proposed methodology. This improved performance is due to  the optimum sparse array smartly engaging  its degrees of freedom to eradicate the interfering signals while maintaining maximum gain towards all  sources of interest.
\subsection{ Fully augmentable linear arrays}
Consider selecting $14$ antennas out of $24$ possible available  locations with   antenna spacing of $\lambda/2$. A desired source is impinging from DOA of $30^0$ and SNR of $10$ dB, whereas  narrowband jammers are operating at $20^0$, $40^0$ and  $120^0$ with INR of $10$ dB each. The range of $\mu$ and other parameters are the same as in \ref{Monte1}. Optimum array configuration (Fig. \ref{op_1D}) achieved through convex relaxation (NFSDR-approach) has an  output SINR of $21.29 $ dB   as compared to SINR of  $21.32$ dB of an optimum array recovered through enumeration ($1.96*10^6$ possible configurations). It should be noted that the array recovered without filled co-array constraint is  not essentially fully augmentable as is the case  in the optimum array (Fig. \ref{op_1D}) which clearly has missing co-array lags. 

In quest of fully augmentable array design we  prefix  8 antennas (red elements in Fig. \ref{mm1_1D-hybrid}) in a minimum redundancy array (MRA)  configuration over 24 uniform grid points. This provides 24 consecutive autocorrelation lags.   We are, therefore, left with six antennas to be placed in the remaining 16 possible locations ($8008$ possible configurations). We  enumerated the performance of all possible hybrid arrays associated with underlying MRA configuration  and found the output SINR ranges from $18.1$ dB to $21.3$ dB. Figure \ref{mm1_1D-hybrid} shows  the  configuration recovered through the proposed approach  which has an output SINR of $20.96$ dB.  The proposed approach thus recovers the hybrid sparse array with performance close to the best possible, moreover it  approximately yields $3$ dB advantage over worst fully augmentable hybrid array. As MRAs are not unique we started with   a different 8 element MRA structured  array (red elements in Fig. \ref{mm2_1D-hybrid}),  to further reinforce the effectiveness of fully augmentable sparse arrays. The dynamic performance range associated with  MRA of Fig. \ref{mm2_1D-hybrid}, is  from $17.59$ dB to   $21.3$ dB. The performance in this case is very similar to the aforementioned MRA configuration  with the output SINR of  $21.08$ dB for the hybrid array recovered through proposed methodology (Fig. \ref{mm2_1D-hybrid}). The maximum possible SINR offered by both hybrid arrays is $21.3$ dB which  is extremely close to SINR performance of $21.32$ dB offered by the optimum array without augmentability constraint.  
\begin{figure}[!t]
	\centering
	\includegraphics[width=3.38 in, height=2.50 in]{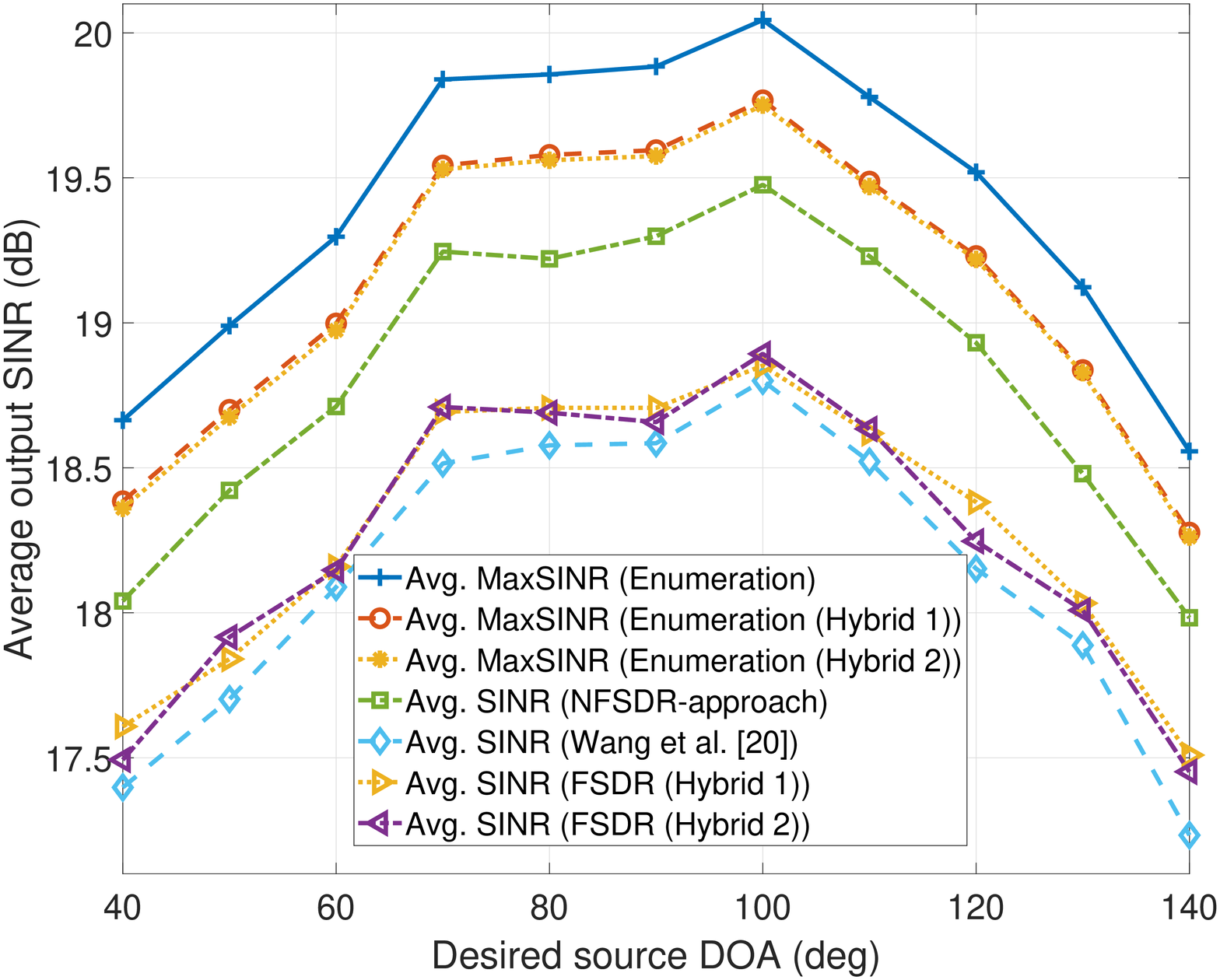}
	\caption{Average Output SINR for different array topologies over $3500$ Monte Carlo trials}
	\label{scan_monte2}
\end{figure}
\subsubsection{Monte Carlo Simulation}
We generate 3500  Monte Carlo simulations for comparison between  the performance of the sparse arrays that are designed freely and that of sparse array design involving   full augmentability constraint.  We choose $16$ antennas out of $24$ available locations.  The desired source DOA is fixed  with SNR of $10$ dB as in \ref{Monte1}.  We assume twelve   narrowband interferences drawn uniformly  from $20^0$ to $160^0$ with  respective INRs   uniformly distributed from $10$~dB to $15$~dB.  For binary search, the upper and lower limit of the sparsity parameter $\mu$ is  5 and 0.01 respectively and $\epsilon=0.1$,  for all 3500 scenarios.  Fig. \ref{scan_monte2} shows the average SINR performance, where the proposed NFSDR-approach  is only $0.57$ dB suboptimal relative to the optimum array found through enumeration (choosing 16 antennas out of 24 involves $735471$ prospective configurations). However,  this performance is achieved by sparse arrays without ensuring the augmentabilty constraint. Therefore, we prefix $8$ antennas in MRA topology, namely Hybrid $1$  and Hybrid $2$ prefix configurations,  shown in red circles in  Figs.  \ref{mm1_1D-hybrid} and Fig. \ref{mm2_1D-hybrid} respectively. The MaxSINR performance, found by enumeration, for either of the  underlying hybrid topologies competes very closely as evident in Fig. 8. The average MaxSINR (found by enumeration), under both prefixed configurations, is only compromised by $0.28$ dB relative to the average MaxSINR performance offered without the augmentability constraint. It is noted that in this case, the possible sparse configurations are drastically reduced from $735471$ to  $12870$ (choose the remaining $8$ antennas from the remaining $16$ possible locations due to prefixing $8$ antennas \textit{a~priori}). It is clear from Fig. \ref{scan_monte2} that the proposed FSDR algorithm successfully  recovers the hybrid sparse array with an average SINR performance loss of $0.8$ dB. We remark that  the performance of the hybrid sparse array  is still slightly better than the optimum sparse array receive beamforming proposed in  \cite{8313065} that assumes the knowledge of jammers' steering vectors and utilizes all the available degrees of freedom, unlike the hybrid sparse array. 

\subsection{ Fully augmentable 2D arrays}
 Consider a $7 \times 7$ planar array with grid pacing of $\lambda/2$ where we place 24 antennas at 49 possible positions.  A desired source is impinging from elevation angle $\theta=50^0$  and azimuth angle of  $\phi=90^0$. Here, elevation angle is with respect to the plane carrying the array  rather than reference from the zenith. Four strong interferes are impinging from  ($\theta=20^0$, $\phi=30^0$), ($\theta=40^0$, $\phi=80^0$), ($\theta=120^0$, $\phi=75^0$) and ($\theta=35^0$, $\phi=20^0$). The INR corresponding to each interference is $20 $ dB and SNR is set to $0$ dB. There are of the order of $10^{14}$ possible 24 antenna configurations, hence the problem is prohibitive even by exhaustive search. Therefore, we resort to  the upper bound of performance limits to compare our results. Here, we utilize the fact that the best possible performance  occurs when the interferes are completely canceled  in the array output and the output SINR in that case would equal the array gain offered by the 24 element array which amounts to  $13.8$~dB.  Figure \ref{2d_ideal} shows the optimum antenna locations recovered by the proposed NFSDR-approach. The output SINR for this configuration is $13.68$ dB which is sufficiently  close to the ideal performance. It should be noted that again the array recovered in the Fig. \ref{2d_ideal} is not fully augmentable  as it is missing quiet a few correlation lags.
 
  We now introduce the condition of full augmentability  by placing 19 antennas in nested lattice configuration \cite{6214623}  to form a filled co-array (red elements in Fig.  \ref{2d_hybridop}). The rest of five available antennas can be placed in the remaining 30 possible locations hence resulting in approximately $1.5 * 10 ^5$ possibilities. Figure \ref{2d_hybridop} shows the  hybrid sparse geometry recovered by FSDR algorithm and offers SINR of $13.25$ dB which is around $0.4$ dB less than the optimum array. The performance range of the hybrid arrays associated with the structured nested lattice array ranges from $11.4$ dB to $13.38$ dB (found through exhaustive search). In this regard the FSDR algorithm finds the  hybrid sparse  array with the performance degradation  of little more than $0.1$ dB. The worst performing hybrid array (Fig.  \ref{2d_hybridwo}) has an output SINR of $11.4$ dB and is around $2$ dB lower than the  best performing hybrid sparse array.
  
   It is of interest to compare the performance of aforementioned sparse arrays with a compact $2$D array. For this purpose, we chose a $6 \times 4$ rectangular array. The compact rectangular array performs very poorly in the underlying scenario and has an output SINR of $7.8$ dB which is more than $ 5$ dB down from the  hybrid  sparse array recovered through the semidefinite relaxation. This performance degradation is very clear from  the beampattern of both arrays shown in Figs. \ref{2d_bp_ideal} and \ref{2d_bp_rec} (normalized beampattern in dB). In the  case of the  hybrid sparse array recovered through FSDR (Fig. \ref{2d_hybridop}), the target has the maximum gain towards the direction of interest with minimum gain simultaneously towards all unwanted DOAs (Fig. \ref{2d_bp_ideal}).  In contrast,  it is clear from Fig. \ref{2d_bp_rec} that the beampattern of the compact rectangular array  could not manage maximum gain towards the direction of interest while effectively rejecting the interfering signals. Although, the $6 \times 5$  and $6 \times 6$ compact arrays utilize $6$ and $12$ additional sensors, yet the  respective output SINRs of $9.04$ dB and  $11$ dB are considerably suboptimal relative to the proposed solutions. It is noted that adding 18  additional sensors resulting in $7 \times 6$ rectangular array has an output SINR of $12.87$ dB. Still, the 24 element free-design as well as the hybrid design outperform the compact $42$ element  rectangular array. However, a $49$ element fully populated $7 \times 7$ rectangular array has an output SINR of $14.37$ dB, which is marginal improvement given the SINR of 24 element designed topologies.  The hybrid array also appears to be more  robust as it has higher dynamic performance range threshold ($11.4$ dB).  The   performance of arbitrarily designed arrays is more prone  to deteriorate catastrophically  even far worse than  that   of the  compact uniform or  rectangular arrays.

  \begin{figure}[!t]
  	\centering
  	\includegraphics[height=2.48in,width=3in]{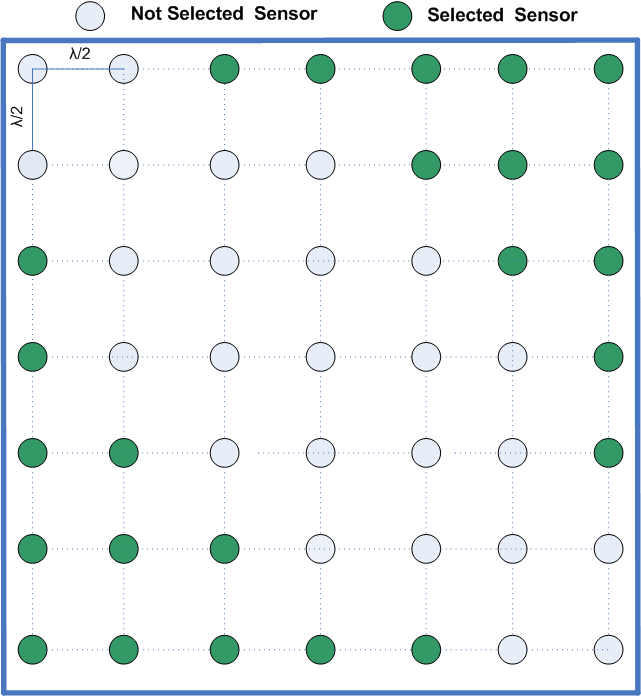}
  	\caption{$24$ element antenna sparse array (NFSDR-approach)}
  	\label{2d_ideal}
  \end{figure}
  We also test the fully augmentable array design for the case of  multiple point source scenario described previously (Section \ref{Multiple point sources}). The hybrid array recovered through proposed methodology is shown in the Fig. \ref{op_alg_multi} (red elements showing the $7$ element MRA). The output SINR is $11.566$ dB and is sufficiently close to the performance achieved through enumeration. 
  \begin{figure}[!t]
  	\centering
  	\includegraphics[height=2.48in,width=3in]{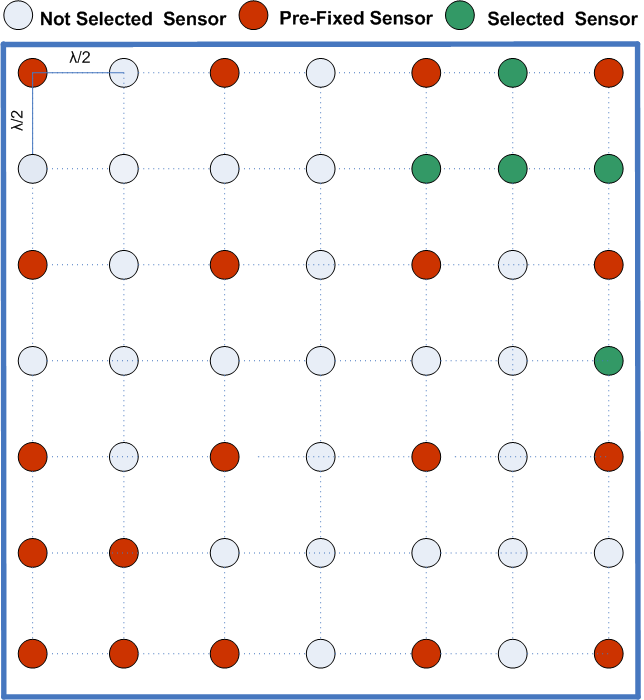}
  	\caption{ $24$ element  hybrid antenna sparse array  ($19$ prefixed, $5$ selected through FSDR)}
  	\label{2d_hybridop}
  \end{figure}
  \begin{figure}[!t]
  	\centering
  	\includegraphics[height=2.48in,width=3in]{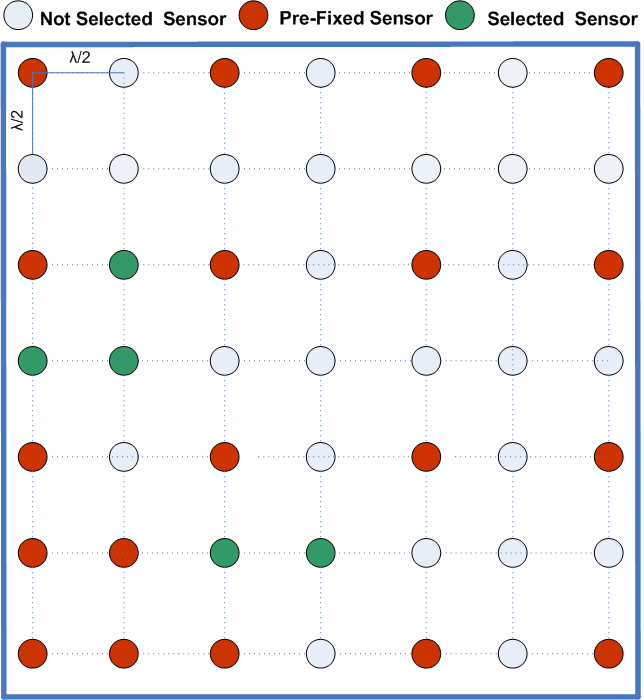}
  	\caption{ $24$ element worst performing hybrid antenna sparse array  ($19$ prefixed, $5$ selected)}
  	\label{2d_hybridwo}
  \end{figure}

\section{Conclusion}
This paper considered fully augmentable  sparse array configurations for maximizing the beamformer output SINR for general rank desired signal  correlation matrices.  It proposed a  hybrid  sparse array design that  simultaneously considers co-array and environment-dependent objectives.  The proposed array design approach uses a subset of the available antennas to obtain a fully augmentable array while employing the remaining antennas  for achieving the highest SINR. It was shown that the hybrid design  is data driven and  hence practically  viable, as it ensures the availability of the full data correlation matrix  with a reasonable trade off in the SINR performance. We applied the modified re-weighting QCQP  which proved effective     in recovering  superior SINR performance for hybrid sparse arrays in polynomial run times. The proposed approach  was  extended  for    fair gain beamforming towards   multiple   sources. We solved the optimization problem by both the proposed algorithms and  enumeration and showed strong agreement between the two methods. 
\begin{figure}[!t]
	\centering
	\includegraphics[width=3.48 in, height=2.55 in]{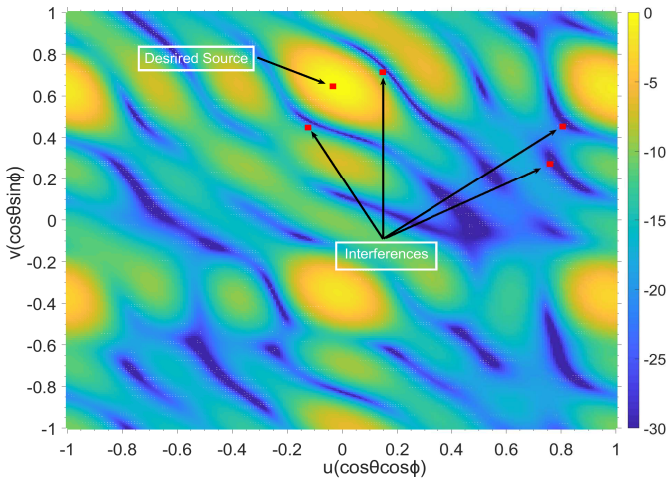}
	\caption{Beampattern  for the antenna array in Fig. \ref{2d_hybridop}}
	\label{2d_bp_ideal}
\end{figure}

\appendix[Proof of the Conjugate symmetric property of optimal weight vector]\label{Appendix1}
\begin{proof}
	The correlation matrix $\mathbf{R}$ for centro-symmetric arrays have a conjugate persymmetric structure such that \cite{86911}: \\
\begin{equation}\label{appendixa}
\mathbf{TR^{'}T}=\mathbf{R}
\end{equation}

Here  $\{'\}$ is the conjugate operator and  $\mathbf{T}$ is the  transformation matrix which flips the entries of a vector upside down by left multiplication;\\ 
\[
\mathbf{T}=
\begin{bmatrix}
	0       & \dots & 0 &0  & 1 \\
	0       & \dots & 0 & 1 & 0 \\
	\vdots     & \dots   &  &\vdots \\
	1       & \dots & 0 &  & 0
\end{bmatrix}
\]
The  optimal weight vector which maximizes the SINR is given by;
\begin{equation}\label{appendixb}
\mathbf{w}_o = \mathscr{P} \{  \mathbf{R_{s^{'}}}^{-1}\mathbf{R_s}  \}
\end{equation}
where,
\begin{equation}\label{appendixc}
\{  \mathbf{R_{s^{'}}}^{-1}\mathbf{R_s} \}  \mathbf{w}_o = \Lambda_{max} \mathbf{w}_o
\end{equation}
Using the relation in  (\ref{appendixa}),    (\ref{appendixc}) can be re-expressed as follows,
 \begin{equation}\label{appendixd}
 \begin{aligned}
 \{  \mathbf{(TR^{'}_{s^{'}}T)^{-1}}\mathbf{(TR^{'}_sT)} \}  \mathbf{w}_o & = \Lambda_{max} \mathbf{w}_o \\
  \{  \mathbf{T^{-1}(R^{'}_{s^{'}})^{-1}T^{-1}}\mathbf{(TR^{'}_sT)} \}  \mathbf{w}_o & = \Lambda_{max} \mathbf{w}_o \\
  \end{aligned}
 \end{equation}
 Multiplying both sides by $\mathbf{T}$ and applying the conjugate operator,
 \begin{equation}\label{appendixe}
 \{  \mathbf{R_{s^{'}}}^{-1}\mathbf{R_s} \}  \mathbf{T}\mathbf{w}^{'}_o = \Lambda_{max} \mathbf{T}\mathbf{w}^{'}_o
 \end{equation}
 From (\ref{appendixe}), we note, that $\mathbf{T}\mathbf{w}_o^{'}$ is also the principal eigenvector associated with matrix $   \mathbf{R_{s^{'}}}^{-1}\mathbf{R_s} $. Since the principal eigenvector of the positive definite hermitian matrix  is unique up to  the scalar complex multiplier, this directly implies that;
  \begin{equation*}\label{appendixf}
 \mathbf{w}_o= \mathbf{T}\mathbf{w}_o^{'} 
 \end{equation*} 
\end{proof}

\begin{figure}[!t]
	\centering
	\includegraphics[width=3.48 in, height=2.55 in]{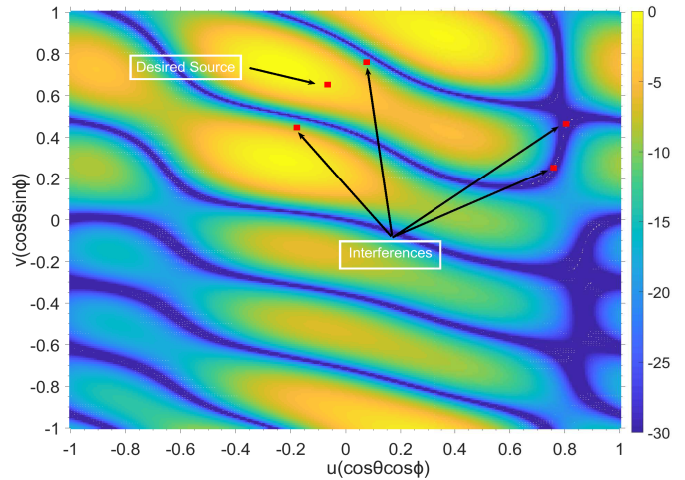}
	\caption{Beampattern  for a $6 \times 4$ compact rectangular  array}
	\label{2d_bp_rec}
\end{figure}

\ifCLASSOPTIONcaptionsoff
  \newpage
\fi

\bibliographystyle{IEEEtran}
\bibliography{references}
\end{document}